\documentclass[pra,twocolumn,showpacs,showkeys,linenumbers]{revtex4}

\usepackage{amssymb,amsfonts,amsmath,amsthm,enumerate}
\usepackage{color,graphicx}
\usepackage{ulem}\renewcommand{\em}{\it}
\usepackage{lipsum}

\newcommand{\Tr}{\mathrm{Tr}\,}
\newcommand{\aux}{\mathrm{aux}}
\newtheorem{proposition}{Theorem}
\newtheorem*{proposition*}{Theorem}
\newtheorem{corollary}{Corollary}
\newtheorem{lemma}{Lemma}

\def\A{\mathcal{A}}
\def\B{\mathcal{B}}
\def\D{\mathcal{D}}
\def\H{\mathcal{H}}
\def\I{\mathcal{I}}
\def\P{\mathcal{P}}
\def\Q{\mathcal{Q}}
\def\U{\mathcal{U}}
\def\Dset{\mathbb{D}}
\def\Rset{\mathbb{R}}
\def\Sset{\mathbb{S}}
\def\lp{\mathrm{LP}}
\def\povm{\mathrm{POVM}}
\def\c{\mathrm{c}}
\def\;{\, ; \,}

\begin{document}

\title{Geometric  approach  to extend  Landau-Pollak  uncertainty relations  for
  positive operator-valued measures}
\author{G.M.\  Bosyk$^{1,2}$,  S.\  Zozor$^{2,1}$,  M.\  Portesi$^{1,2}$,  T.M.\
  Os\'an$^3$, and P.W.\ Lamberti$^3$\vspace{2mm}}
\affiliation{$^1$  Instituto   de  F\'isica   La  Plata  (IFLP),   CONICET,  and
  Departamento de  F\'isica, Facultad de Ciencias  Exactas, Universidad Nacional
  de La Plata, 1900 La Plata, Argentina, \vspace{2mm}}
\affiliation{$^2$ Laboratoire Grenoblois  d'Image, Parole, Signal et Automatique
  (GIPSA-Lab, CNRS),  11 rue des Math\'ematiques, 38402  Saint Martin d'H\`eres,
  France, \vspace{2mm}}

\affiliation{$^3$  Facultad de  Matem\'atica, Astronom\'ia  y  F\'isica (FaMAF),
  Universidad Nacional  de C\'ordoba, and  CONICET, Avenida Medina  Allende S/N,
  Ciudad Universitaria, X5000HUA, C\'ordoba, Argentina\vspace{2mm}}


\begin{abstract}

  We  provide a twofold  extension of  Landau--Pollak uncertainty  relations for
  mixed quantum states and for positive operator-valued measures, by recourse to
  geometric considerations. The generalization  is based on metrics between pure
  states,  having the form  of a  function of  the square  of the  inner product
  between the states.  The triangle inequality satisfied by such metrics plays a
  crucial role in our derivation.  The usual Landau--Pollak inequality is thus a
  particular case (derived  from Wootters metric) of the  family of inequalities
  obtained,  and, moreover, we  show that  it is  the most  restrictive relation
  within the family.


\end{abstract}

\date{\today}

\maketitle


\section{Introduction}
\label{sec:introduction}

The uncertainty  principle is  one of the  most important principles  in quantum
mechanics.   Originally  stated  by  Heisenberg~\cite{Hei27}, it  establishes  a
limitation  on  the  simultaneous  predictability of  incompatible  observables.
Uncertainty relations constitute the quantitative expressions of this principle.
The     first    formulations,    due     to    Heisenberg,     Robertson    and
Schr\"odinger~\cite{Hei27,  VURs}, are  the  most popular  ones.  However,  they
exhibit a state-dependent lower bound for the product of the variances of a pair
of noncommuting observables.  Thus, a drawback of this formulation is that it is
not universal;  moreover, the universal bound  (the minimum over  the states) is
trivial, i.e., it equals zero.   Several alternatives have been studied, such as
those using the  sum of the variances instead  of the product~\cite{SumVURs}, or
those using information-theoretic measures as quantifiers of lack of information
(see  Refs.\ \cite{MaaUff88,  improveMU,  EURs}, and  \cite{reviews} for  recent
reviews).  Recently,  some of us~\cite{ZozBos14}  extended entropic formulations
of  the  uncertainty  principle to  the  case  of  a  pair of  observables  with
nondegenerate discrete $N$-dimensional  spectra, using generalized informational
entropies.    The  proposed   formulation  makes   use  of   the  Landau--Pollak
inequality~(LPI)    which     has    been    introduced     in    time-frequency
analysis~\cite{LanPol61},  and  later  on  adapted  to  the  quantum  mechanics'
language~\cite{MaaUff88}.   In  this  last  case,  the inequality  is  itself  a
(geometric) formulation of the uncertainty principle and applies for pure states
and nondegenerate observables. The goal of this paper is precisely to extend the
scope of Landau--Pollak-type inequalities.

Let us  consider a pure state  $| \Psi \rangle$ belonging  to an $N$-dimensional
Hilbert space $\H$, and two observables with discrete, nondegenerate spectra and
corresponding  eigenbases  $\{  |a_i\rangle  \}_{i  = 1,  \ldots,  N}$  and  $\{
|b_j\rangle \}_{j = 1, \ldots, N}$, described by the operators $A = \sum_{i=1}^N
a_i\ |a_i\rangle \langle  a_i|$ and $B = \sum_{j=1}^N  b_j\ |b_j \rangle \langle
b_j|$. The LPI~\cite{LanPol61, MaaUff88} reads:
\begin{equation}
\label{eq:LPI}
\arccos \! \left( \max_i |\langle a_i | \Psi \rangle| \right) + \arccos \!
\left( \max_j |\langle b_j | \Psi \rangle| \right) \geq \arccos c
\end{equation}
where $c  = \max_{ij}  |\langle a_i |  b_j \rangle|  $ is the  so-called overlap
between   the  eigenbases  of   the  observables.    The  overlap   ranges  from
$\frac{1}{\sqrt N}$  (complementary observables) to 1 (the  observables share at
least an  eigenstate).  The  quantity $|\langle a_i  | \Psi  \rangle|^2$ (resp.\
$|\langle  b_j  | \Psi  \rangle|^2$)  is  interpreted  as the  probability  that
observable $A$ (resp.\  $B$), for a system in  preparation $|\Psi\rangle$, takes
the  eigenvalue $a_i$ (resp.\  $b_j$).  The  LPI~\eqref{eq:LPI} is  a meaningful
expression of the strong uncertainty principle for nondegenerate observables and
pure states.  On  the one hand, unlike standard approaches,  the right hand side
of~\eqref{eq:LPI}  is  a  state-independent   lower  bound  on  the  probability
distributions associated with  observables $A$ and $B$. The  bound is nontrivial
whenever $c < 1$,  i.e., when the observables $A$ and $B$  do not share a common
eigenstate. On the other hand, when $\max_i |\langle a_i | \Psi \rangle|^2 = 1$,
i.e.,  when  the  probability  distribution  associated  to  observable  $A$  is
concentrated, LPI implies that $\max_i |\langle b_i | \Psi \rangle|^2 \leq c^2 <
1$ that means  that the probability distribution associated  with observable $B$
cannot   be  concentrated;  Furthermore,   in  the   complementary  case   $c  =
\frac{1}{\sqrt{N}}$, the probability distribution associated with observable $B$
must be uniform.  Besides, LPI has been used to improve Maassen--Uffink entropic
uncertainty  relation~\cite{improveMU}, and  a  weak version  has  been used  to
obtain entanglement criteria~\cite{MiyIma07, entLP}.

The  extraction of  information from  a  quantum system  inevitably requires  to
perform a  measurement. The simplest one  is a projection  valued measure (PVM),
also known as  von Neumann measurement. However, a  description of a measurement
by PVM is  in general insufficient because most of the  observations that can be
performed  are not  of this  type. Besides,  for many  applications one  is only
interested in  the probability distributions  associated to the  observables but
not  in  the  post-measurement   state.   Fortunately,  there  exists  a  simple
mathematical tool  known as  positive operator-valued measures  (POVM) formalism
which provides a  generalization of standard projective measurements  and also a
description       of       any       possible      measurement       to       be
performed~\cite{NieChu10,BenZyc06}.

Until  recently,  LPI  had  only  been demonstrated  for  pure  quantum  states.
However, in a recent contribution~\cite{BosOsa14}, inequality~\eqref{eq:LPI} was
generalized to deal with mixed states for the case of nondegenerate observables,
and  was  indeed  extended (based  on  geometric  concepts)  using a  family  of
uncertainty  measures other  than the  arccosine (which  is related  to Wootters
metric).  Here  we go a step  further, extending the  LPI~\eqref{eq:LPI} and its
generalization given  in~Ref.~\cite{BosOsa14}, in order to  deal with degenerate
observables described by POVM sets.  In Sec.~\ref{sec:GLPI}, we first resume the
framework, namely  that of two  observables described by  two POVM sets,  and we
describe  measures  of uncertainty  based  on a  class  of  metric between  pure
states. Then we formulate our main  results, namely: (i) extension of the LPI in
the  context of POVM  using a  class of  generalized uncertainty  measures, (ii)
determination of the  most restrictive measure within the  class considered, and
(iii)  analysis  of  the  uncertainty   intrinsic  to  a  given  POVM  set.   In
Sec.~\ref{sec:Illustrations} we provide some numerical illustrations. We analyze
the consequences of the extended LPI in the context of POVM, and we also discuss
the optimality or  not of the uncertainty relations  obtained.  Some conclusions
are drawn in Sec.~\ref{sec:Conclusions}.  The  proof of the extended LPI is made
in  several steps,  developed  in detail  in  App.~\ref{app:Proofs}; whereas  in
App.~\ref{app:simulations}   the  algorithms  used   for  the   simulations  are
presented.


\section{Generalized Landau--Pollak inequalities: main results}
\label{sec:GLPI}

Let us consider two observables $A$ and  $B$ described by the POVM sets $\A = \{
A_i \}_{i=  1,\ldots,N_A}$ and $\B = \{  B_j \}_{j=1,\ldots,N_B}$, respectively,
i.e.,  $\A$  and   $\B$  are  sets  of  Hermitian   (or  self-adjoint)  positive
semi-definite operators  acting on an  $N$-dimensional Hilbert space  $\H$, that
satisfy   the   completeness   relation   or   resolution   of   the   identity:
$\sum_{i=1}^{N_A} A_i  = I  = \sum_{j=1}^{N_B} B_j$,  where $I$ is  the identity
operator on $\H$, and $N_A$ and  $N_B$ are not necessarily equal to one another,
or equal to $N$.  In some sense, the $A_i$'s (resp.\ $B_j$'s) allow to represent
the possible  outcomes of observable $A$  (resp.\ $B$). Also, let  us consider a
system whose  state is described  by a density  operator $\rho$ acting  on $\H$,
where $\rho$ is  Hermitian, positive semi-definite, and normalized  ($\Tr \rho =
1$). We denote by $\D$ the set of density operators.  The quantity
\begin{equation*}
p_i(A;\rho) = \Tr(A_i \rho)
\end{equation*}
represents the probability of measuring the $i$th outcome of $A$ when the system
is in  the state  $\rho$ ~\cite{BenZyc06}.  In  the context of  observables with
nondegenerate spectra and  pure states, the operators take the  form of rank one
projectors,      $A_i       =      |a_i\rangle      \langle       a_i|$      and
$\rho=|\Psi\rangle\langle\Psi|$.

\

Let us now turn to the consideration  of a measure of uncertainty that allows us
for the generalization of the LPI.  We start with continuous functions $f: [0 \;
1] \mapsto  \Rset_+$, that are strictly  decreasing and satisfy $f(1)  = 0$, and
such  that  for two  (normalized)  pure  states $|  \Psi  \rangle$  and $|  \Phi
\rangle$, \ $f \left( \left| \langle \Psi | \Phi \rangle \right|^2 \right) = d_f
\left( | \Psi \rangle  \, , \, | \Phi \rangle \right)$  defines a metric between
them.  This kind  of metrics is interesting as they depend  on the inner product
between  two  quantum  states  and,  hence, they  are  invariant  under  unitary
transformations. Some well-known cases are:
\begin{itemize}
\item  $f(x)  = \arccos\sqrt{x}$,  leading  to  the  Wootters metric,  or  Bures
  angle~\cite{WoottersMetric},
\item   $f(x)   =   \sqrt{2    \,   (1-\sqrt{x})}$,   leading   to   the   Bures
  metric~\cite{Bur69, BenZyc06},
\item $f(x) = \sqrt{1-x}$, related to the root-infidelity metric~\cite{GilLan05},
  or Hilbert--Schmidt or trace distance~\cite{BenZyc06}.
\end{itemize}
We notice  that these  metrics extend  to (or, indeed,  were defined  for) mixed
states,  with function  $f$  being applied  to  the fidelity  between two  mixed
states~\cite{BenZyc06,fidelity}.  From such metrics $d_f$, the quantity
\begin{equation}
\U_f \left( \A ; \rho \right) \, = \, f\left( P_{\A;\rho} \right)
\label{eq:Uf}
\end{equation}
with
\begin{equation}
P_{\A;\rho} = \max_i \, \Tr(A_i \rho) = \max_i \, p_i(A;\rho)
\end{equation}
defines an uncertainty measure corresponding to the measurement of a set $\A$ of
operators that describe  observable $A$, for a system in a  state $\rho$, in the
sense that~\cite{BosOsa14}
\begin{itemize}
\item  $\U_f\left(  \A   ;  \rho  \right)  \ge  0$   for  all  $P_{\A;\rho}  \in
  \left[\frac{1}{N_A},1\right]$, and
\item  $\U_f\left( \A  ; \rho  \right)$ is  decreasing in  terms of  the maximal
  probability $P_{\A;\rho}$, with
  \begin{itemize}
  \item[$\diamond$] $\U_f\left( \A ; \rho \right)$ is maximum iff $P_{\A;\rho} =
    \frac{1}{N_A}$,  that   is  equivalent  to   the  equiprobability  situation
    $p_i(A;\rho) = \frac{1}{N_A}$ for all $i$,
  \item[$\diamond$] $\U_f\left(  \A ; \rho \right)$ vanishes  iff $P_{\A;\rho} =
    1$,  that   is  equivalent  to   the  certainty  situation   $p_i(A;\rho)  =
    \delta_{ik}$ for a given $k$.
  \end{itemize}
\end{itemize}

\

Our  main result in  the present  contribution is  a two-fold  generalization of
Landau--Pollak-type uncertainty relations, comprising  the cases of mixed states
and of  POVM descriptions.  We establish  the following theorem,  whose proof is
postponed until Apps.~\ref{subsec:Projectors_pure}--\ref{subsec:POVM_mixed}, and
give a discussion below.
\begin{proposition}
  Let $\A = \{ A_i \}_{i = 1, \ldots,  N_A}$ and $\B = \{ B_j \}_{j = 1, \ldots,
    N_B}$  be   two  positive  operator  valued   measures  describing  discrete
  observables  $A$  and $B$,  respectively,  and  acting  on an  $N$-dimensional
  Hilbert  space $\H$.  Then  for an  arbitrary density  operator $\rho  \in \D$
  acting on $\H$, the following relation holds:
\begin{equation}
\U_f\left( \A ; \rho \right) + \U_f\left( \B ; \rho \right) \ge f \left( c^{\
2}_{\A,\B} \right)
\label{eq:LPI_mixed_POVM}
\end{equation}
where
\begin{equation}
c_{\A,\B} = \max_{ij} \left\| \sqrt{A_i} \, \sqrt{B_j} \right\|\ = \ \max_{ij}
\left\| \sqrt{B_j} \, \sqrt{A_i} \right\|
\label{eq:overlap}
\end{equation}
is the generalized overlap between the two POVM sets.
\label{prop:LPI_mixed_POVM}
\end{proposition}
\noindent The overlap~\eqref{eq:overlap} is given  in terms of an operator norm.
For the  sake of completeness, we  recall here its definition:  for any operator
$O$ on $\H$, \ $\displaystyle \| O \| = \max_{|\varphi\rangle \in \H} \frac{\| O
  |\varphi\rangle  \|}{\| |\varphi\rangle  \|} =  \max_{|\varphi\rangle  \in \H}
\frac{\langle \varphi | O^\dag  O | \varphi \rangle^{\frac12}}{\langle \varphi |
  \varphi \rangle^{\frac12}}  = \max_{|\Psi\rangle \in \H: \|  |\Psi\rangle \| =
  1}  \|  O  |\Psi\rangle  \|$,  where  $O^\dag$  is  the  adjoint  operator  of
$O$~\cite{ReeSim80-Con90}.

Notice that, in the case of observables with nondegenerate spectra, letting $A_i
= |a_i\rangle  \langle a_i|$ and $B_j  = |b_j\rangle \langle b_j|$  then one has
$c_{\A,\B} = \max_{ij}  | \langle a_i | b_j \rangle |  = c$.  The generalization
of LPI  to mixed  states proved in  Ref.~\cite{BosOsa14} is then  recovered from
Theorem~\ref{prop:LPI_mixed_POVM}.       Moreover,       in      this      case,
inequality~\eqref{eq:LPI_mixed_POVM} is  sharp whatever  $f$, in the  sense that
there  exists at least  one state  that renders  equality.  Indeed,  denoting by
$(i',j')$  the  pair  of  indices  such  that  $c_{\A,\B}  =  \|  \sqrt{A_{i'}}\
\sqrt{B_{j'}} \|$, and choosing $|\Psi\rangle = |a_{i'}\rangle$ or $|\Psi\rangle
= |b_{j'}\rangle$, together  with the fact that $f(1) = 0$,  allows to prove the
assertion.

A way to look at  the family of inequalities~\eqref{eq:LPI_mixed_POVM} is in the
sense that they  establish a restriction to the values that  the pair of maximal
probabilities $\left( P_{\A;\rho} \, , \, P_{\B;\rho} \right)$ can take jointly,
within  the   rectangle  $\left[  \frac{1}{N_A}  \;  1   \right]  \times  \left[
  \frac{1}{N_B} \;  1 \right]$.   We point  out that the  fact discussed  in the
preceding paragraph for the context  of nondegenerate observables, does not mean
that the whole family of inequalities  renders the same permitted domain for the
pair;  that  is neither  the  case  in the  POVM  context.   And,  in fact,  the
restriction   imposed  by~\eqref{eq:LPI_mixed_POVM}   manifests  in   a  reduced
rectangle. Indeed, if $P_{\A;\rho} \le c_{\A,\B}^{\ 2}$ (resp.\ $P_{\B;\rho} \le
c_{\A,\B}^{\  2}$),   then  $f(P_{\A;\rho})  \ge   f(c_{\A,\B}^{\  2})$  [resp.\
$f(P_{\B;\rho})       \ge       f(c_{\A,\B}^{\       2})$~]       and       thus
inequality~\eqref{eq:LPI_mixed_POVM}  is  satisfied  whatever  $P_{\B;\rho}  \in
\left[  \frac{1}{N_B}   \;  1  \right]$  \big(resp.\   $P_{\A;\rho}  \in  \left[
  \frac{1}{N_A}  \;   1  \right]$\big).    But  if  (and   only  if)   the  pair
$(P_{\A;\rho},P_{\B;\rho})$ is within the square $(c_{\A,\B}^{\ 2} \; 1]^2$, the
inequality  becomes restrictive.   A  careful look  at~\eqref{eq:LPI_mixed_POVM}
suggests us to define the function
\begin{equation}
\begin{array}{ccl}
  h_c^f & : & [c^2 \; 1]  \to  [ c^2 \; 1 ] \\ [2mm]
  h_c^f (x) & = &  f^{-1} \left( f(c^2) - f(x) \right)
\end{array}
\label{eq:hcf}
\end{equation}
\noindent  Thus,  the  restriction due  to  inequality~\eqref{eq:LPI_mixed_POVM}
writes down as
\begin{equation}
P_{\B;\rho} \le h_{c_{\A,\B}}^f \big( P_{\A;\rho} \big) \qquad \mbox{for} \qquad
P_{\A;\rho} \in (c_{\A,\B}^{\ 2} \; 1]
\label{eq:DomainPaPb_mixed_POVM}
\end{equation}
(a similar relation is valid  exchanging the roles  of $\A$ and  $\B$). 

In  the case of  nondegenerate spectra,  the fact  that the  lower bound  to the
uncertainty   sum  can   be  reached   whatever   $f$,  is   evidenced  in   the
maximum-probabilities plane in the fact  that the points $(c^2,1)$ and $(1,c^2)$
coincide   for  all   curves   $y   =  h_c^f(x)$,   as   already  mentioned   in
Ref.~{\cite{BosOsa14},  and there  do  exist states  for  which $(P_{\A;\rho}  ,
  P_{\B;\rho}) =  (1,c^2)$ and  those for which  $(P_{\A;\rho} ,  P_{\B;\rho}) =
  (c^2,1)$.  The question that had remained open, was to know which function $f$
  of    the    family    considered    leads    to    the    most    restrictive
  inequality~\eqref{eq:DomainPaPb_mixed_POVM},   i.e.,   which   $f$   minimizes
  $h_c^f(P)$ when $c$ and $P \in (c^2  \; 1)$ are fixed.  The answer is given in
  the   following    theorem,   the   proof    of   which   is    presented   in
  App.~\ref{subsec:Wootters}:
\begin{proposition}
  Within   the    whole   family   of   uncertainty    inequalities   given   by
  Theorem~\ref{prop:LPI_mixed_POVM}, the  strongest restriction for  the pair of
  maximal probabilities $\left( P_{\A;\rho} , P_{\B;\rho} \right)$, rewritten as
  inequality~\eqref{eq:DomainPaPb_mixed_POVM}, and its counterpart changing $\A$
  with  $\B$, corresponds  to Wootters  case, namely  for the  function  $f(x) =
  \arccos\sqrt{x}$.
\label{prop:Wootters}
\end{proposition}

It is important  also to address the following  situation: when considering only
one observable, in the general POVM context, there exists a possible uncertainty
that is intrinsic to the POVM representation itself \cite{Mas07}.  Indeed a POVM
set  $\A$ can  be such  that whatever  the state  of the  system is,  no outcome
appears with certainty.   This situation arises when no operator  in $\A$ has an
eigenvalue  equal  to  unity.   An  inequality  quantifying  such  an  intrinsic
uncertainty     is      given     in     the      following     corollary     to
Theorem~\ref{prop:LPI_mixed_POVM},     whose    proof     is     presented    in
App.~\ref{subsec:proof_Corollary1}:
\begin{corollary}
  Let  $\A =  \{  A_i  \}_{i= 1,  \ldots,  N_A}$ be  a  POVM  set describing  an
  observable $A$, and acting on an $N$-dimensional Hilbert space $\H$.  Then for
  an  arbitrary density operator  $\rho \in  \D$ acting  on $\H$,  the following
  relation holds:
\begin{equation}
\U_f\left( \A ; \rho \right) \ge f \left( c^{\ 2}_\A \right) ,
\label{eq:Intrinsic_POVM}
\end{equation}
where
\begin{equation}
c_\A = \ \max_i \left\| \sqrt{A_i} \right\| \ \in \ \left[\frac{1}{\sqrt{N_A}},
1 \right]
\label{eq:overlapPOVM}
\end{equation}
is a  generalized intrinsic overlap  of the POVM  set.  The bound  is nontrivial
(only) when the eigenvalues of any operator $A_i$ are different from unity.
\label{prop:Intrinsic_POVM}
\end{corollary}
Combining                  Theorem~\ref{prop:LPI_mixed_POVM}                 and
Corollary~\ref{prop:Intrinsic_POVM}, it appears that the lower bound for the sum
of  metric-based uncertainties  of  the form~\eqref{eq:Uf}  can  be improved  as
follows:
\begin{corollary}
  Let $\A = \{ A_i \}_{i = 1, \ldots,  N_A}$ and $\B = \{ B_j \}_{j = 1, \ldots,
    N_B}$ be two POVM sets  describing observables $A$ and $B$ respectively, and
  acting  on an  $N$-dimensional  Hilbert  space $\H$.   Then  for an  arbitrary
  density operator $\rho \in \D$ acting on $\H$, the following relation holds:
\begin{equation}
\U_f\left( \A ; \rho \right) + \U_f\left( \B ; \rho \right) \ge \max \left\{ f
\left( c_\A^{\ 2} \right) + f \left(c_\B^{\ 2} \right),  f \left( c^{\ 2}_{\A,\B}
\right) \right\}
\label{eq:LPI_mixed_POVM_improved}
\end{equation}
where $c_\A$,  $c_\B$ and  $c_{\A,\B}$ are the  intrinsic and  joint generalized
overlaps.
\label{prop:LPI_mixed_POVM_improved}
\end{corollary}

We notice  that the overlap $c_{\A,\B}$ is  bounded from below and  above in the
following way (see App.~\ref{subsec:proof_bounds_c} for the proof):
\begin{equation}
\max \left\{\ \frac{c_{\A}}{\sqrt{N_B}} \, , \, \frac{c_{\B}}{\sqrt{N_A}}\
\right\} \: \le \: c_{\A,\B} \: \le \: c_\A\ c_\B
\label{eq:Domain_cab}
\end{equation}
Consequently, from  the last inequality and  since $c_\A c_\B  \leq \min \left\{
  c_\A , \ c_\B \right\}$, we  get $f(c_{\A,\B}^{\ 2}) \ge f \left( \min \left\{
    c_\A^{\ 2} \ , \ c_\B^{\ 2} \right\} \right) = \max \left\{ f(c_{\A}^{\ 2})\
  , \ f(c_\B^{\ 2}) \right\} $.  When at least one operator $A_i$ (and/or $B_j$)
of the POVM set $\A$ (and/or of  $\B$) has an eigenvalue equal to unity, $c_\A =
1$  (and/or $c_\B  = 1$),  thus we  get $f(c_{\A,\B}^{\  2}) \ge  f(c_\B^{\ 2})$
(and/or  $f(c_{\A,\B}^{\ 2})  \ge f(c_\A^{\  2})$):  together with  $f(1) =  0$,
$f(c_{\A,\B}^{\  2}) \ge  f(c_\A^{\ 2})  + f(c_\B^{\  2})$, i.e.,  the  bound in
Eq.~\eqref{eq:LPI_mixed_POVM_improved}       reduces       to      that       of
Eq.~\eqref{eq:LPI_mixed_POVM}.

\ 

It  turns out  that the  allowed domain  for the  pair of  maximal probabilities
$(P_{\A;\rho}  ,  P_{\B;\rho})$  is   constrained  as  given  in  the  following
corollary, the proof of which is given in App.~\ref{subsec:proof_Corollary3}:
\begin{corollary}
  Let $\A = \{ A_i \}_{i = 1, \ldots,  N_A}$ and $\B = \{ B_j \}_{j = 1, \ldots,
    N_B}$ be two POVM sets  describing observables $A$ and $B$ respectively, and
  acting  on an  $N$-dimensional  Hilbert  space $\H$.   Then  for an  arbitrary
  density  operator $\rho$  acting on  $\H$, the  pair of  maximal probabilities
  $(P_{\A;\rho} , P_{\B;\rho})$ is constrained to the domain
\begin{widetext}
\begin{equation}
\Dset_\lp(c_\A,c_\B,c_{\A,\B}) = \left\{ (P_\A,P_\B) \in \left[ \frac{1}{N_A} \;
c_\A^{\ 2} \right] \times \left[ \frac{1}{N_B} \; c_\B^{\ 2} \right]\ :\ P_\B
\le h_{c_{\A,\B}}\big(P_\A\big) \: \mbox{when} \: P_\A \ge c_{\A,\B}^{\ 2}
\right\}
\label{eq:DomainPaPb_Intrinsic_mixed_POVM}
\end{equation}
\end{widetext}
where $h_{c_{\A,\B}}$ is given by \eqref{eq:hcf} for $f(x)=\arccos\sqrt x$.
If $c_\B^{\ 2} \le h_{c_{\A,\B}}(c_\A^{\  2})$, i.e., $c_{\A,\B} \ge c_\A c_\B -
\sqrt{(1-c_\A)(1-c_\B)}$,  the allowed domain  becomes $\left[  \frac{1}{N_A} \;
  c_\A^{\ 2} \right] \times \left[ \frac{1}{N_B} \; c_\B^{\ 2} \right]$.
\label{prop:DomainPaPb_Intrinsic_mixed_POVM}
\end{corollary}

\

Let  us  now illustrate  both  theorems by  simulated  POVM  sets and  simulated
states. These simulations allow us  to comment on the uncertainty relations and,
in   particular,   on  Corollary~\ref{prop:DomainPaPb_Intrinsic_mixed_POVM}   in
various contexts.


\section{Numerical illustrations}
\label{sec:Illustrations}

This section  aims at illustrating  the constraints imposed on  the simultaneous
predictability   of   two   observables   as  expressed   by   the   uncertainty
relations~\eqref{eq:LPI_mixed_POVM}  in  Theorem~\ref{prop:LPI_mixed_POVM}.   To
this end,  we draw randomly several POVM  pairs; and for any  given pair $(\A_k,
\B_k)$  of POVM ($k  = 1,  \ldots$), we  draw randomly  mixed states  $\{ \rho_l
\}_{l=1,\ldots}$. Then  we calculate  the uncertainty sums  $\U_f(\A_k;\rho_l) +
\U_f(\B_k;\rho_l)$ and  the corresponding bounds $ f  \left( c^{\ 2}_{\A_k,\B_k}
\right)$  for   different  functions~$f$  (see   App.~\ref{app:simulations}  for
technical details on the simulation of POVM and states).  In order to illustrate
Corollary~\ref{prop:DomainPaPb_Intrinsic_mixed_POVM},   we   analyze  not   only
$\U_f(\A_k;\rho_l) + \U_f(\B_k;\rho_l)$, but also the cloud of points $\{ \left(
  P_{\A_k;\rho_l}  \,  ,   \,  P_{\B_k;\rho_l}  \right)  \}_{l=1,\ldots}$  where
$(\A_k,\B_k)$    is    fixed,     together    with    their    allowed    domain
$\Dset_\lp(c_{\A_k},c_{\B_k},c_{\A_k,\B_k})$.

Figures~\ref{fig:uncertainty_nondegenerate}.(a)--(c)  represent the simultaneous
uncertainty   $\U_f(\A_k;\rho_l)  +   \U_f(\B_k;\rho_l)$   versus  the   overlap
$c_{\A_k,\B_k}$, compared to the lower bound $f\left(c_{\A_k,\B_k}^{\ 2}\right)$
for: (a)  Wootters metric  given by $f(x)  = \arccos\sqrt{x}$, (b)  Bures metric
with  $f(x) =  \sqrt{2 \,  (1-\sqrt{x})}$, and  (c) root-infidelity  metric with
$f(x) = \sqrt{1-x}$,  in the context of observables  with nondegenerate spectra,
and  for both  pure and  mixed states.  Here, the  operators written  as  $A_i =
|a_i\rangle\langle a_i|,  i = 1, \ldots,  N$, are built from  the column vectors
$|a_i\rangle$ of a unitary matrix  (and similarly for the $B_j$).  The dimension
is    chosen     to    be    $N     =    3$.     These     figures    illustrate
Theorem~\ref{prop:LPI_mixed_POVM} and the fact  that, in the nondegenerate case,
the        bounds         that        we        find         are        optimal.
Figure~\ref{fig:uncertainty_nondegenerate}.(d)      depicts      the      domain
$\Dset_\lp(1,1,0.75)$  and   functions  $h_{c_{\A,\B}}^f$  for   the  Bures  and
root-infidelity metrics, together with  snapshots of pairs $\left( P_{\A;\rho_l}
  \,  ,  \, P_{\B;\rho_l}  \right)$:  this  clearly  illustrates that  the  case
corresponding to  Wootters metric gives  the most restrictive domain  within the
family  of  uncertainty  inequalities  (Theorem~\ref{prop:Wootters}).   It  also
suggests that, in  the nondegenerate context, $\Dset_\lp$ is  the best domain in
the sense  that it  coincides with  $\{ \left( P_{\A;\rho}  \, ,  \, P_{\B;\rho}
\right): \rho \in \D \}$.  This assertion remains however to be proved.

\begin{figure*}
\includegraphics[width=\textwidth]{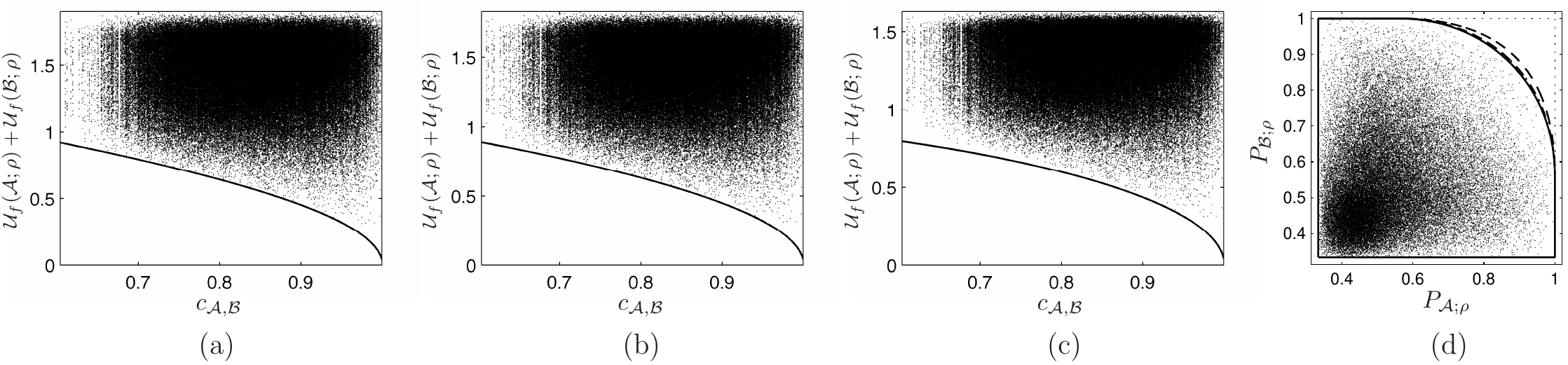}
\caption{Illustration     of    the     uncertainty    relations     given    in
  Theorem~\ref{prop:LPI_mixed_POVM} in the case of nondegenerate observables and
  $N   =  3$.    Snapshots   of  the   uncertainty   sum  $\U_f(\A_k;\rho_l)   +
  \U_f(\B_k;\rho_l)$  vs.~the corresponding generalized  overlap $c_{\A_k,\B_k}$
  (points), and comparison  to the bound $f(c_{\A,\B}^{\ 2})$  (solid line) for:
  (a)~$f(x) =  \arccos\sqrt{x}$ (Wootters); (b)~$f(x)  = \sqrt{2\ (1-\sqrt{x})}$
  (Bures); (c)~$f(x) = \sqrt{1-x}$ (root-infidelity). In the simulation, $k = 1,
  \ldots, 10^4$ and $l = 1, \ldots, 25$.  (d)~Domain $ \Dset_\lp(1,1,0.75)$ that
  corresponds to  Wootters metric (solid line),  and functions $h_{c_{\A,\B}}^f$
  for the Bures (dashed-dotted  line) and root-infidelity (dashed line) metrics;
  the points represent  the pairs $(P_{\A,\rho_l},P_{\B,\rho_l})$ with $(\A,\B)$
  fixed and $l = 1, \ldots, 5 \times 10^4$.}
\label{fig:uncertainty_nondegenerate}
\end{figure*}

Figure~\ref{fig:uncertainty_Dlp}    depicts    some    examples    of    domains
$\Dset_\lp(c_\A,c_\B,c_{\A,\B})$,
Eq.~(\ref{eq:DomainPaPb_Intrinsic_mixed_POVM}), together with snapshots of pairs
$(P_{\A;\rho_l},P_{\B;\rho_l})$, in various  contexts. The dimensions chosen are
$N  = 3$,  $N_A  = 4$  and  $N_B =  5$.  In  Figs.~\ref{fig:uncertainty_Dlp}.(a)
and~(b),   $c_{\A,\B}  <  c_\A   c_\B  -   \sqrt{(1-c_\A)  (1-c_\B)}$,   and  in
Figs.~\ref{fig:uncertainty_Dlp}.(c)   and~(d)   $c_{\A,\B}   >   c_\A   c_\B   -
\sqrt{(1-c_\A) (1-c_\B)}$,  illustrating situations where  $\Dset_\lp$ restricts
to $\left[  \frac{1}{N_A} \; c_\A^{\  2} \right] \times \left[  \frac{1}{N_B} \;
  c_\B^{\ 2}  \right]$ or  not.  In Figs.~\ref{fig:uncertainty_Dlp}.(a)  and (b)
[resp.\  Figs.~\ref{fig:uncertainty_Dlp}.(c) and (d)],  the overlaps  are equal,
but the pairs of  POVM are different.  It can be seen  that the domain where the
pairs $(P_{\A_k;\rho_l} \,  , \, P_{\B_k;\rho_l})$ live does  not depend only on
the overlaps, but also  depend on the pair of POVM itself.   To be more precise,
dealing with optimality, two notions have to be considered:
\begin{itemize}
\item $(\A,\B)-$optimal  domain $\Dset_\povm(\A,\B) = \left\{  (P_{\A;\rho} \, ,
    \, P_{\B;\rho}): \rho \in \D\right\}$  is the smallest domain containing all
  pairs $(P_{\A;\rho}  \, , \, P_{\B;\rho})$  for any mixed state  $\rho \in \D$
  acting on $\H$.
\item         $(c_\A,c_\B,c_{\A,\B})-$optimal         set         $\displaystyle
  \Dset_\c(c_\A,c_\B,c_{\A,\B})  = \mathop{\bigcup_{(\A,\B) \mbox{\tiny\  s.t. }
      (c_\A,c_\B,c_{\A,\B})}}   \Dset_\povm(\A,\B)$,  is   the  union   of  sets
  $\Dset_\povm(\A,\B)$    with    $(\A,\B)$    sharing    the    same    triplet
  $(c_\A,c_\B,c_{\A,\B})$ of overlaps.
\end{itemize}
These  two  domains  are  probably  not  equal.  Moreover,  at  this  step,  the
illustrations seem to indicate that  $\Dset_\lp \ne \Dset_\povm$, but it remains
to  be  proved  formally.   The  $(c_\A,c_\B,c_{\A,\B})-$optimality  or  not  of
$\Dset_\lp$  remains also  to be  investigated,  but the  illustrations seem  to
indicate that this last optimality is probable.

\begin{figure*}
\includegraphics[width=\textwidth]{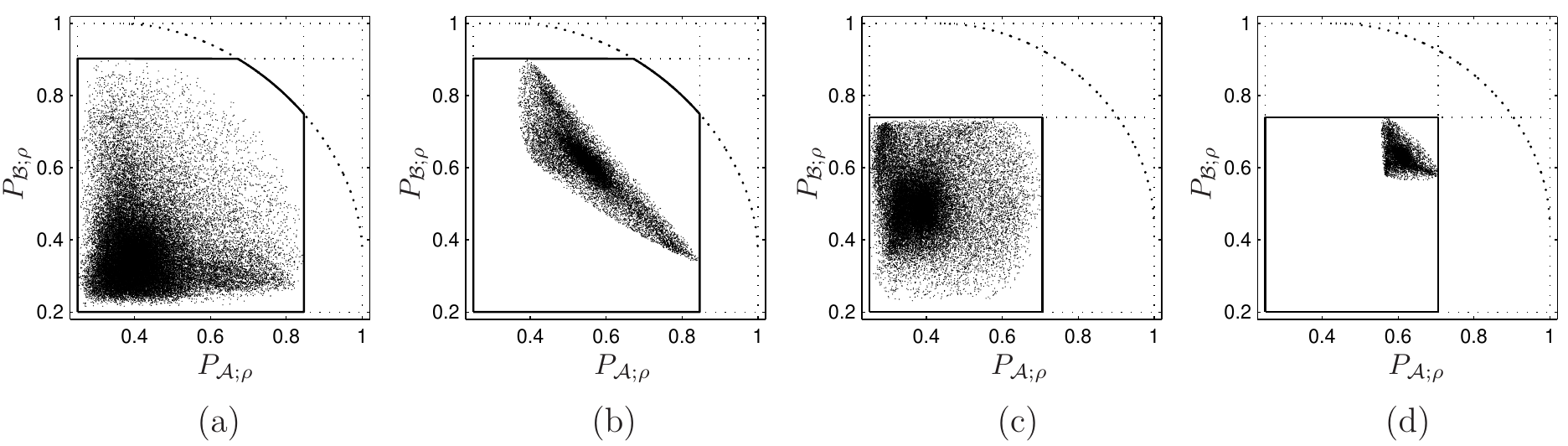}
\caption{Illustration  of  Corollary~\ref{prop:DomainPaPb_Intrinsic_mixed_POVM}:
  domain $\Dset_\lp(c_\A,c_\B,c_{\A,\B})$ and snapshots of pairs $(P_{\A;\rho_l}
  \, , \, P_{\B;\rho_l})$  (points).  In each figure, $N= 3$, $N_A  = 4$, $N_B =
  5$, and the dotted lines represent respectively $P_{\A;\rho} = \frac{1}{N_A}$,
  $P_{\B;\rho} = \frac{1}{N_B}$, $P_{\A;\rho} = c_\A^2$, $P_{\B;\rho} = c_\B^2$,
  and     $P_{\B;\rho}     =     h_{c_{\A,\B}}(P_{\A;\rho})$.     The     domain
  $\Dset_\lp(c_\A,c_\B,c_{\A,\B})$ is  delimited by the  solid line. In  (a) and
  (b): $(c_\A  \, , \, c_\B  \, , \, c_{\A,\B}  ) = (0.92 \,  , \, 0.95  \, , \,
  0.60)$; in (c) and (d):  $(c_\A \, , \, c_\B \, , \, c_{\A,\B}  ) = (0.84 \, ,
  \, 0.86 \, , \, 0.84)$.  The number of snapshots is of the order of $10^4$.}
\label{fig:uncertainty_Dlp}
\end{figure*}


\section{Concluding remarks}
\label{sec:Conclusions}

In  summary,  in this  work  we  have been  able  to  extend the  Landau--Pollak
inequality for degenerate observables described by the most general formalism of
positive  operator-valued   measures.   We  derived  a   family  of  uncertainty
relations,  given  in  Theorem~\ref{prop:LPI_mixed_POVM},  in the  most  general
context of observables described by POVM sets and for mixed quantum states.  The
relations   obtained  extend  and   generalize  the   well-known  Landau--Pollak
inequality in  the general context and  provide a whole  family of inequalities.
The starting point that gives rise  to this set of relations is the assimilation
of  measures of uncertainty  in terms  of a  conveniently defined  metric, which
satisfies  the triangle  inequality.  We  adopt metrics  that lie  on decreasing
functions of the square of the  inner product between pure states.  It comes out
that  Wootters  metric, leading  to  the  usual  Landau--Pollak inequality  (its
extension to mixed states and  POVM descriptions) is the most restrictive within
the family of inequalities  we obtain (Theorem~\ref{prop:Wootters}).  From these
theorems, we recover that in  general, in the POVM representation context, there
exists    an    uncertainty    intrinsic    to   the    representation    itself
(Corollary~\ref{prop:Intrinsic_POVM}), and  thus, that the  allowable domain for
the  pair   of  maximal  probabilities  corresponding  to   two  observables  is
constrained  by  both the  joint  uncertainty  relation  and the  intrinsic  one
(Corollary~\ref{prop:LPI_mixed_POVM_improved}).

A  direct consequence of  our results  is that  a previous  work \cite{ZozBos14}
dealing with  generalized entropies of probability vectors  extends very easily
in the most general case of POVM representations of observables.

Finally, the simulated  numerical results suggest that for a  given pair of POVM
$\A$ and $\B$, the allowable  domain for the pair $(P_{\A;\rho},P_{\B;\rho})$ is
tighter     that     domain     $\Dset_\lp(c_\A,c_\B,c_{\A,\B})$    given     by
Corollary~\ref{prop:DomainPaPb_Intrinsic_mixed_POVM}:  the questions  of finding
the tightest  domain for $(P_{\A;\rho},P_{\B;\rho})$,  given POVM sets  $\A$ and
$\B$     [$(\A,\B)-$    optimal     set    $\Dset_\povm(\A,\B)$]     or    given
$(c_\A,c_\B,c_{\A,\B})$        [$(c_\A,c_\B,c_{\A,\B})-$       optimal       set
$\Dset_\c(c_\A,c_\B,c_{\A,\B})$],  remain  open.  The  structure  of  the  tight
domains  (convex or  not?)  and the  properties  of the  states  (pure or  not?)
reaching the border of these domains are also open questions.  These points give
possible directions for further investigation in the field.

\acknowledgments

SZ  and MP  are very  grateful to  the R\'egion  Rh\^one-Alpes (France)  for the
grants that  enabled this work.  MP  and GMB also  acknowledge financial support
from CONICET (Argentina), and warm  hospitality during their stays at GIPSA-Lab.
PWL and  TMO are  grateful to SECyT-UNC  (Argentina) for financial  support.  We
warmly thank  Pr.\ Pierre Comon from  GIPSA-lab for the  useful discussion about
the simulation of sets of self-adjoint semidefinite matrices adding up to unity.



\appendix
\section{Proofs of the theorems and corollaries}
\label{app:Proofs}


\subsection{Landau--Pollak-type uncertainty relation  for sets of projectors and
  pure states}
\label{subsec:Projectors_pure}

The first  step in our demonstration  is to consider  pure states $|\Psi\rangle$
and sets of projectors $\P = \{ P_i \}_{i = 1, \ldots, N_P}$, being $P_i^{\,2} =
P_i$.   The case  of PVM  sets is  a particular  case, where  $P_i P_{i'}  = P_i
\delta_{ii'}$.   The  following  lemma  extends  the LPI  for  these  particular
measurements.

\begin{lemma}
  Let $\P = \{ P_i \}_{i = 1, \ldots,  N_P}$ and $\Q = \{ Q_j \}_{j = 1, \ldots,
    N_Q}$ be two  sets of projectors acting on  an $N$-dimensional Hilbert space
  $\H$.  Then for  an arbitrary pure state $|\Psi\rangle  \in \H$, the following
  relation holds:
\begin{equation}
\U_f\left( \P ; | \Psi \rangle \langle \Psi | \right) \, + \, \U_f\left( \Q ; |
\Psi \rangle \langle \Psi | \right) \: \geq \: f(c^{\ 2}_{\P,\Q})
\label{eq:LPI_pure_POVM_proj}
\end{equation}
where  $\displaystyle c_{\P,\Q}  =  \max_{ij} \left\|  \sqrt{P_i} \,  \sqrt{Q_j}
\right\|\ =\ \max_{ij} \left\| \sqrt{Q_j} \, \sqrt{P_i} \right\|$.
\label{prop:LPI_pure_POVM_proj}
\end{lemma}
\begin{proof}
  Note first  that for any  operator $O$, the  operator norm satisfies  $\|O\| =
  \|O^\dag\|$ \cite{ReeSim80-Con90}.  This  property together with the Hermitian
  property  of  operators  $P_i$  and  $Q_j$  justifies  the  equality  $\left\|
    \sqrt{P_i}  \,  \sqrt{Q_j}  \right\|  =  \left\|  \sqrt{Q_j}  \,  \sqrt{P_i}
  \right\|$.  Consider now the two normalized pure states:
\begin{equation}
|\psi_i\rangle = \frac{P_i |\Psi\rangle}{\| P_i |\Psi\rangle \|} \quad
\mathrm{and} \quad |\varphi_j\rangle = \frac{Q_j |\Psi\rangle}{\| Q_j
|\Psi\rangle \|} 
\end{equation}
for $\| P_i |\Psi\rangle  \| \ne 0$ and $\| Q_j |\Psi\rangle  \| \ne 0$, then 
\begin{equation}
|\langle \Psi | \psi_i \rangle|^2 = \langle \Psi | P_i | \Psi \rangle
\quad \mathrm{and} \quad
|\langle \Psi|\varphi_j\rangle|^2 = \langle \Psi | Q_j | \Psi \rangle
\label{eq:projPsi}
\end{equation}
using that $P_i^{\,2} = P_i$ and $Q_j^{\,2} = Q_j$. 

Now,  the triangle  inequality  fulfilled by  the  metric $d_f$  applied to  the
triplet $| \psi_i \rangle, | \varphi_j \rangle$ and $| \Psi \rangle$ reads:
\begin{equation}
\label{eq:LPI_POVM_projectors_pure_state}
f\left(\langle \Psi | P_i | \Psi \rangle \right) \, + \, f \left( \langle \Psi |
Q_j | \Psi \rangle \right) \: \geq \: f \left( \frac{|\langle \Psi| P_i Q_j |
\Psi\rangle|^2}{\| P_i | \Psi \rangle\|^2 \| Q_j | \Psi \rangle\|^2} \right).
\end{equation}
Notice that
\begin{eqnarray}
\left| \langle \Psi| P_i Q_j | \Psi \rangle \right| & = &
\left| \langle \Psi |\ \sqrt{P_i}\ \sqrt{P_i}\ \sqrt{Q_j}\
\sqrt{Q_j}\ | \Psi \rangle \right|\nonumber\\[2mm]
& \leq & \left\| \sqrt{P_i}\ | \Psi \rangle \right\| \
\left\| \sqrt{P_i}\ \sqrt{Q_j}\ \sqrt{Q_j}\ | \Psi\rangle
\right\|\nonumber\\[2mm]
& \leq & \left \| \sqrt{P_i}\ | \Psi \rangle \right\|\ \left\| \sqrt{Q_j}\ |
\Psi \rangle \right\| \, \left\| \sqrt{P_i}\ \sqrt{Q_j} \right\|\nonumber\\[2mm]
& \leq & \left \| \sqrt{P_i}\ | \Psi \rangle \right\|\ \left\| \sqrt{Q_j}\ |
\Psi \rangle \right\| \, c_{\P,\Q}\label{eq:CS},
\end{eqnarray}
where  the first inequality  follows from  the Cauchy--Schwartz  inequality, the
second one from the definition of the  operator norm, and the third one from the
definition of  $c_{\P,\Q}$.  The proof ends  noting that $\sqrt{P_i}  = P_i$ and
$\sqrt{Q_j} = Q_j$, choosing $i'$ and $j'$ so that $\langle \Psi | P_{i'} | \Psi
\rangle =  \max_i \langle \Psi |  P_i | \Psi \rangle$  (necessarily nonzero) and
$\langle  \Psi |  Q_{j'} |  \Psi \rangle  =  \max_j \langle  \Psi |  Q_j |  \Psi
\rangle$,  together with the  decreasing property  of the  function $f$  and the
definition  of  $\U_f$ given  by  Eq.~\eqref{eq:Uf}.  Let  us mention  that  the
interchange of the roles  of $P_i$ and $Q_j$ leads to the  same result due to $|
\langle \psi_i | \varphi_j \rangle|^2  = |\langle \varphi_j | \psi_i \rangle|^2$
and the symmetry satisfied by the operator norm.
\end{proof}

Note that the  sets $\P$ and $\Q$ do  not need to satisfy the  resolution of the
identity,  i.e.,  the  inequality  applies  beyond the  scope  of  the  complete
description of observables by sets of projectors.


\subsection{Landau--Pollak-type uncertainty relation for POVM pairs and pure
  states}
\label{subsec:POVM_pure}

We can extend now the previous result to general  POVM sets.
\begin{lemma}
  Let $\A = \{ A_i \}_{i = 1, \ldots,  N_A}$ and $\B = \{ B_j \}_{j = 1, \ldots,
    N_B}$ be two  POVM sets describing observables $A$ and $B$  and acting on an
  $N$-dimensional  Hilbert  space  $\H$.   Then  for  an  arbitrary  pure  state
  $|\Psi\rangle \in \H$, the following relation holds:
\begin{equation}
\U_f\left( \A ; | \Psi \rangle \langle \Psi | \right) \, + \, \U_f\left( \B ; |
\Psi \rangle \langle \Psi | \right) \: \geq \: f(c^{\ 2}_{\A,\B})
\label{eq:LPI_pure_POVM}
\end{equation}
where  $\displaystyle c_{\A,\B}  =  \max_{ij} \left\|  \sqrt{A_i} \,  \sqrt{B_j}
\right\|\ =\ \max_{ij} \left\| \sqrt{B_j} \, \sqrt{A_i} \right\|$.
\label{prop:LPI_pure_POVM}
\end{lemma}

\begin{proof}
  Let us  consider the  pure state  $| \Phi \rangle  = |  \Psi \rangle  \oplus 0
  \oplus 0$ belonging to the extended  Hilbert space which is the direct sum $\H
  \oplus  \H^\aux \oplus  \H^\aux$, where  $\H^\aux$ has  the same  dimension as
  $\H$.  Consider also projectors $P_i$ and $Q_j$ of the form~\cite{MiyIma07}
\begin{eqnarray*}
P_i &=& \begin{pmatrix}
A_i & \sqrt{A_i ( I - A_i)} & 0 \\[2.5mm]
\sqrt{A_i ( I - A_i)} & I - A_i & 0\\[2.5mm]
0 & 0 & 0
\end{pmatrix},\\[2.5mm]
Q_j &=& \begin{pmatrix}
B_j & 0 & \sqrt{B_j (I - B_j)} \\[2.5mm]
0 & 0 & 0  \\[2.5mm]
\sqrt{B_j (I - B_j)} & 0 & I - B_j
\end{pmatrix}.
\end{eqnarray*}
Using that
\begin{eqnarray*}
\langle \Phi | P_i | \Phi \rangle &=& \langle \Psi | A_i | \Psi \rangle ,\\
\langle \Phi | Q_j | \Phi \rangle &=& \langle \Psi | B_j | \Psi \rangle ,\\
\frac{|\langle \Phi | P_i Q_j | \Phi \rangle|}{\| P_i | \Phi \rangle\| \| Q_j |
\Phi \rangle\|} &=& \frac{|\langle \Psi| A_i B_j | \Psi \rangle|}{\| \sqrt{A_i}
| \Psi \rangle\| \| \sqrt{B_j} | \Psi \rangle\|} ,
\end{eqnarray*}
then inequality~\eqref{eq:LPI_POVM_projectors_pure_state}  applied to any triplet
$P_i, Q_j$ and $|\Phi\rangle$  so that $\langle \Phi | P_i |  \Phi \rangle \ne 0
\ne \langle \Phi | Q_j | \Phi \rangle$, leads to
\begin{eqnarray}
\label{eq:LPI_POVM_pure_state}
&& f \left( \langle \Psi | A_i | \Psi \rangle \right) \, + \, f \left( \langle
\Psi | B_j | \Psi \rangle \right) \: \nonumber \\
&& \qquad \qquad \geq \: f \left( \frac{|\langle \Psi| A_i
B_j | \Psi \rangle|^2}{\| \sqrt{A_i} | \Psi \rangle\|^2 \| \sqrt{B_j} | \Psi
\rangle\|^2} \right). 
\end{eqnarray}
Now,  as done in~\eqref{eq:CS},  we have  $\frac{|\langle \Psi|  A_i B_j  | \Psi
  \rangle|}{\| \sqrt{A_i} |  \Psi \rangle\| \| \sqrt{B_j} |  \Psi \rangle\|} \le
c_{\A,\B}$.    The   end    of   the    proof    is   similar    to   that    of
Lemma~\ref{prop:LPI_pure_POVM_proj}.
\end{proof}

Note  that, here  again, the  sets $\A$  and  $\B$ do  not need  to fulfill  the
resolution of  the identity.  Thus, Lemma~\ref{prop:LPI_pure_POVM}  applies in a
more general context than that of the description of observables by POVM.


\subsection{Landau--Pollak-type uncertainty relation for POVM sets and mixed
  states}
\label{subsec:POVM_mixed}

We   are   ready    now   to   prove   inequality~\eqref{eq:LPI_mixed_POVM}   in
Theorem~\ref{prop:LPI_mixed_POVM}.  To  this  end,  let us  consider  a  density
operator  $\rho$   acting  on  $\H$.    Since  it  is  Hermitian   and  positive
semidefinite, it can be diagonalized  on an orthonormal basis $\{|l\rangle\}$ of
$\H$, i.e.,  $ \rho =  \sum_{l = 1}^N  \rho_l \, | l  \rangle \langle l  |$ with
$\rho_l \ge  0$ and  $\sum_l \rho_l =  \Tr \rho  = 1$.  Let  us then  consider a
purification $ | \Phi' \rangle$ of  $\rho$, belonging to a product Hilbert space
$\H \otimes \widetilde{\H}^\aux$,
\begin{equation}
|\Phi' \rangle = \sum_{l=1}^N \sqrt{\rho_l} \ | l \rangle \otimes |
l^{\mathrm{aux}} \rangle,
\end{equation}
where $\{  | l^{\mathrm{aux}} \rangle \}$  is an arbitrary  orthonormal basis of
$\widetilde{\H}^\aux$    (without     loss    of    generality,     we    assume
$\widetilde{\H}^\aux$ of the  same dimension as $\H$).  The  mixed state on $\H$
is recovered by the partial trace, that is
\begin{equation}
\Tr_\aux \left( |\Phi' \rangle \langle \Phi' | \right) = \sum_{l=1}^N \rho_l
| l \rangle \langle l | = \rho.
\end{equation}
It can be verified that
\begin{equation*}
\langle \Phi' | A_i \otimes I| \Phi' \rangle = \Tr \left( A_i \rho \right),
\quad
\langle \Phi' | B_j \otimes I | \Phi' \rangle = \Tr \left( B_j \rho \right) ,
\end{equation*}
and that
\begin{eqnarray*}
\left \| \left( \sqrt{A_i \otimes I} \right) \, \left( \sqrt{B_j \otimes I}
\right) \right\| &=& \left\| \left( \sqrt{A_i} \otimes I \right) \, \left(
\sqrt{B_j} \otimes I \right) \right\|\\
= \left\| \left( \sqrt{A_i} \, \sqrt{B_j} \right) \otimes I \right\|
&=& \left\| \sqrt{A_i} \, \sqrt{B_j} \right\| .
\end{eqnarray*}
Applying  inequality~\eqref{eq:LPI_pure_POVM} to  the triplet  $A_i  \otimes I$,
$B_j      \otimes      I$      and      $|\Phi'     \rangle$,      leads      to
inequality~\eqref{eq:LPI_mixed_POVM},    that    concludes    the    proof    of
Theorem~\ref{prop:LPI_mixed_POVM}.


\subsection{Proof of Theorem~\ref{prop:Wootters}: Wootters metric gives
  the most restrictive domain for $(P_{\A;\rho},P_{\B;\rho})$}
\label{subsec:Wootters}

The  inner  product  defines the  cosine  of  an  angle  between two  states  of
$\H$. Thus, in the context of pure states, since $P_{\A;\rho}$ and $P_{\B;\rho}$
are closely linked to inner products,  it can be intuitively guessed that within
the family  of inequalities~\eqref{eq:LPI_mixed_POVM}, the  most restrictive one
occurs    when    $f    =    \arccos\sqrt{x}$.     Indeed,    in    this    case
inequality~\eqref{eq:LPI_mixed_POVM} links the angles between the possible pairs
among    three    vectors   of    $\H$.     In    the    general   context    of
Theorem~\ref{prop:LPI_mixed_POVM}, this guess turns out to be true.

\

First  of all, recall  that inequality~\eqref{eq:LPI_mixed_POVM}  is restrictive
only when  the pair $(P_{\A;\rho},P_{\B;\rho})$ belongs to  $(c_{\A,\B}^{\ 2} \;
1]^2$ under the form~\eqref{eq:DomainPaPb_mixed_POVM}:
\begin{equation}
P_{\B;\rho} \le h_{c_{\A,\B}}^f \big( P_{\A;\rho} \big) \qquad \mbox{for} \qquad
P_{\A;\rho} \in [c_{\A,\B}^{\ 2} \; 1] 
\nonumber
\end{equation}
where $h_c^f(x)  = f^{-1}  \big( f(c^2)  - f(x) \big)$.  Note now  that Wootters
metric, given by $f(x) = \arccos\sqrt{x}$, leads to $h_c(x) = \big( c \sqrt{x} +
\sqrt{1-c^2}  \sqrt{1-x}   \big)^2$  \  (the  superscript   $\arccos$  has  been
suppressed for the sake of simplicity).  Thus, denoting
\begin{equation*}
\gamma = \arccos c \qquad \mbox{and}  \qquad x = \cos^2 \theta \quad \mbox{with}
\quad \theta \in [0 \; \gamma],
\end{equation*}
we can write
\begin{equation*}
h_c(\cos^2 \theta) = \cos^2(\gamma-\theta) .
\end{equation*}

Fix now a decreasing function $f$ and assume that there is a $\theta_f \in [0 \;
\gamma]$ such that
\begin{equation*}
h_c^f(\cos^2 \theta_f) < h_c(\cos^2 \theta_f) = \cos^2(\gamma-\theta_f) . 
\end{equation*}
From  the definition~\eqref{eq:hcf} of  $h_c^f$ and  the decreasing  property of
$f$, this inequality becomes
\begin{equation}
f \big( \cos^2 \theta_f \big) + f \big( \cos^2(\gamma-\theta_f) \big) < f(\cos^2
\gamma)
\label{eq:f_cos}
\end{equation}
Let  us   then  consider  two   orthogonal  pure  states   $|\psi_1\rangle$  and
$|\Psi\rangle$ of $\H$ and let us define the pure states:
\begin{eqnarray*}
|\phi\rangle & = & \cos \theta_f \, |\psi_1\rangle + \sin \theta_f \,
|\Psi\rangle,\\
|\psi_2\rangle & = & \cos \gamma \, |\psi_1\rangle + \sin \gamma \, |\Psi\rangle .
\end{eqnarray*}
It can be  verified that $|\langle \phi | \psi_1  \rangle|^2 = \cos^2 \theta_f$,
$|\langle  \phi | \psi_2  \rangle|^2 =  \cos^2 (\gamma-\theta_f)$  and $|\langle
\psi_2    |     \psi_1    \rangle|^2    =    \cos^2     \gamma$,    such    that
inequality~\eqref{eq:f_cos} writes
\begin{equation*}
f \big( |\langle \phi | \psi_1 \rangle|^2\big) + f \big( |\langle \phi | \psi_2
\rangle|^2\big) < f \big( |\langle \psi_2 | \psi_1 \rangle|^2\big)
\end{equation*}
For such  a function $f$, $d_f  \big( |\Phi\rangle,|\Psi\rangle \big)  = f \big(
|\langle\Phi  |\Psi\rangle|^2 \big)$ is  {\em not}  a metric  since it  does not
satisfy the triangle inequality.  In conclusion, for any function $f$ defining a
metric  $d_f$ between pure  states, $h_c^f(x)  \ge h_c(x)  $ in  $[c^2 \;  1 ]$,
proving  that Wootters  metric  gives  the most  restrictive  inequality of  the
family~\eqref{eq:LPI_mixed_POVM}, as stated in Theorem~\ref{prop:Wootters}.


\subsection{Proof of  Corollary~\ref{prop:Intrinsic_POVM}: Intrinsic uncertainty
  relation for POVM representations}
\label{subsec:proof_Corollary1}

The     proof     of     \eqref{eq:Intrinsic_POVM}     is     immediate     from
Theorem~\ref{prop:LPI_mixed_POVM}, taking $\B = \I \equiv \{ I \}$.  It can also
be proved directly by noting that $p_i (A ; |\Psi\rangle \langle\Psi|) = \langle
\Psi | A_i | \Psi \rangle \le  \| |\Psi\rangle \|\ \| A_i |\Psi\rangle \| \le \|
|\Psi\rangle  \|^2\  \|  A_i  \|  =   \|  A_i  \|$  for  any  (normalized)  pure
state.  Writing a mixed  state as  a convex  combination of  pure-states density
matrices, we get again $p_i (A ; \rho) \le \| A_i \| = \| \sqrt{A_i} \|^2$.  The
proof ends by  choosing the index $i'$ that maximizes $p_i  (A ; \rho)$ together
with the decreasing property of $f$.

Regarding the  interval for the intrinsic  overlap in~\eqref{eq:overlapPOVM}, by
definition  $c_\A^{\  2} \,  \ge  \, \left\|  \sqrt{A_i}  \right\|^2  \, \ge  \,
\left\langle  \Psi  \left| A_i  \right|  \Psi  \right\rangle$  for any  $i$  and
normalized state $|\Psi\rangle$.   Summing over $i$, from the  resolution of the
identity and  since $|\Psi\rangle$ is normalized,  it yields $N_A  \, c_\A^{\ 2}
\ge 1$.  Moreover, $\left\langle \Psi  \left| A_i \right| \Psi \right\rangle \le
1$ implies that $\left\| \sqrt{A_i} \right\|  \le 1$ for all $i$, and thus $c_\A
\le 1$.



\subsection{Proof    of    the    bounds    for    the    overlap    $c_{\A,\B}$
  in~\eqref{eq:Domain_cab}}
\label{subsec:proof_bounds_c}

The first inequality in~\eqref{eq:Domain_cab} comes from $c_{\A,\B}^{\ 2} \, \ge
\, \left\| \sqrt{A_i}\ \sqrt{B_j} \right\|^2  \, \ge \, \left\langle \Psi \left|
    \sqrt{A_i}\ B_j\  \sqrt{A_i} \right| \Psi  \right\rangle$ for any  $i,j$ and
unitary $|\Psi\rangle$.  Summing over $j$,  from the resolution of the identity,
and taking  the maximum  over $|\Psi\rangle$  and then over  $i$, leads  to $N_B
c_{\A,\B}^{\ 2} \ge c_\A^{\ 2}$.  By inverting $A_i$ and $B_j$ in the expression
of $c_{\A,\B}$ the overlap also  satisfies $N_A c_{\A,\B}^{\ 2} \ge c_\B^{\ 2}$.
The  second  inequality  is  a  consequence of  the  submultiplicative  property
$\left\| \sqrt{A_i}\ \sqrt{B_j}  \right\| \le \left\| \sqrt{A_i} \right\|\left\|
  \sqrt{B_j} \right\| \le c_\A\ c_\B$ (see Ref.~\cite{ReeSim80-Con90}).



\subsection{Proof    of    Corollary~\ref{prop:DomainPaPb_Intrinsic_mixed_POVM}:
  Allowed  domain  for  the   pair  of  maximal  probabilities  $(P_{\A;\rho}  ,
  P_{\B;\rho})$}
\label{subsec:proof_Corollary3}

The     proof    is     immediate     from    Theorems~\ref{prop:LPI_mixed_POVM}
and~\ref{prop:Wootters},  Corollary~\ref{prop:LPI_mixed_POVM_improved},  and the
definition $h_c(x) = \cos^2\left(\arccos c - \arccos\sqrt{x} \right) = \left(
  c \sqrt{x} + \sqrt{1-c^2} \sqrt{1-x} \right)^2$.\\
By symmetry, the  roles of $\A$ and $\B$ can  be interchanged, leading naturally
to the same domain.



\section{Simulation of states and of POVM}
\label{app:simulations}

Here, we present the algorithms used in Sec.~\ref{sec:Illustrations} to simulate
quantum states (pure or mixed) and POVM.

Pure   states   can   be    simulated   as   $|\Psi\rangle   =   \Phi(\vartheta)
\frac{|\varphi\rangle}{\| |\varphi\rangle  \|}$ where $\frac{|\varphi\rangle}{\|
  |\varphi\rangle \|}$ has  a uniform distribution on the  unit sphere $\Sset^N$
by drawing $|\varphi\rangle$ according to a zero-mean Gaussian law with identity
covariance matrix~\cite{Knu98}; $\Phi(\vartheta)$ is a diagonal matrix of phases
$\mbox{e}^{\imath \vartheta_i}$ where the $\vartheta_i$ ($i = 1, \ldots, N$) are
mutually independent and  uniformly distributed on $\left[ 0  \; 2 \pi \right)$,
and independent of $|\varphi\rangle$.

In  order to  simulate mixed  states, we  can use  the fact  that  an Hermitian,
positive  semidefinite operator  can be  diagonalized on  an  orthonormal basis,
$\rho = \sum_{m=1}^N \alpha_m |\Psi_m\rangle \langle\Psi_m|$ where $\alpha_m \ge
0$ are  the eigenvalues  of $\rho$, with  $\sum_m \alpha_m  = 1$ because  of the
normalization  of  $\rho$~\cite{ZycSom03}.  Thus,  we  can simulate  orthonormal
bases $\left\{ |\Psi_m\rangle  \right\}_{m = 1, \ldots, N}$ as  the columns of a
randomly drawn unitary matrix~\cite{RandomU}; the coefficients $\alpha_m$ can be
drawn independently according to a uniform  law on $[0 \; 1]$, and normalized to
add to  unity.  Another way of making  should be to generate  a complex Gaussian
random  matrix $M$ and  to compute  $\rho =  \frac{M M^\dag}{\Tr(M  M^\dag)}$ as
proposed for instance in Ref.~\cite{ZycSom03}, or  from a pure state in a higher
dimensional space and taking the partial trace (see App.~\ref{subsec:POVM_mixed}
and Refs.~\cite{ZycSom03, ZycSom05}).

As far  as we  know, there are  no ways  to simulate POVM  sets with  a specific
distribution.  For $\A$ (and similarly  for $\B$), a simple approach may consist
in  drawing a  unitary matrix  $U$~\cite{RandomU} and  a set  of  $N_A$ diagonal
matrices $D_i$ of positive elements, and  to consider the set of matrices $A_i =
U  \left(  \sum_{j=1}^{N_A} D_j  \right)^{-1}  D_i  U^\dag$  that satisfies  the
resolution  of identity.   In  this case  the  $A_i$ give  a  resolution to  the
identity,  but they  share  the same  eigenspace.   To avoid  this drawback,  we
simulate sets of  $N_A$ self-adjoint matrices in the following  way, that can be
viewed as an extension of the previous approach~\cite{Com14}. Let
\begin{equation*}
A_i = R_i\ U_i\ \Delta_i\ U_i^\dag\ R_i^\dag \quad \mbox{for} \quad i = 1,
\ldots, N_A-1
\end{equation*}
and
\begin{equation*}
A_{N_A} = R_{N_A-1}\ U_{N_A-1}\ (I-\Delta_{N_A-1})\ U_{N_A-1}^\dag\ R_{N_A-1}^\dag
\end{equation*}
where
\begin{itemize}
\item  $U_i$ are  unitary matrices  independently  drawn according  to the  Haar
  (uniform) distribution on the set of unitary matrices~\cite{RandomU},
\item $\Delta_i$  are diagonal matrices, where the  components are independently
  drawn according to a uniform distribution on $[0 \; 1]$,
\item    $R_1   =    I$,    and   $R_i    =   R_{i-1}\    U_{i-1}\
  \sqrt{I-\Delta_{i-1}}\quad $ for $i = 2, \ldots, N_A-1$.
\end{itemize}
It can be verified recursively that the $A_i$ form a resolution of the identity,
evaluating the sum  $\left( (A_{N_A} + A_{N_A-1}) +  A_{N_A-2}\right) + \ldots +
A_1$ step by step. Moreover, the $A_i$ do not share the same eigenspace.



\end{document}